\documentclass[runningheads]{llncs}

\usepackage{amsmath}
\usepackage{amsfonts}
\usepackage{amssymb}

\usepackage[pdfborder={0 0 0}]{hyperref}

\newcommand{\df}{\emph}

\def\set#1{{\{ #1 \}}}

\newcommand{\MyGamma}{\mathrm{\Gamma}}
\newcommand{\MyDelta}{\mathrm{\Delta}}

\newcommand{\MyPi}{\mathrm{\Pi}}
\newcommand{\MySigma}{\mathrm{\Sigma}}
\newcommand{\MyOmega}{\mathrm{\Omega}}

\newcommand{\A}{\mathcal{A}}
\newcommand{\aalph}{\{a\}}
\newcommand{\StackAlph}{\MyGamma}
\newcommand{\confto}{\vdash}

\newcommand{\Terminals}{\MySigma}
\newcommand{\Nonterminals}{\MyDelta}
\newcommand{\Productions}{\pi}

\newcommand{\cclass}[1]{\mathbf{#1}}
\renewcommand{\P}{\cclass{P}}
\renewcommand{\L}{\cclass{L}}
\newcommand{\NP}{\cclass{NP}}
\newcommand{\NL}{\cclass{NL}}
\newcommand{\coNP}{\cclass{coNP}}
\newcommand{\SigmaTwoP}{\cclass{\MySigma_2 P}}
\newcommand{\PiTwoP}{\cclass{\MyPi_2 P}}
\newcommand{\PSPACE}{\cclass{PSPACE}}

\newcommand{\NC}{\cclass{NC}}
\newcommand{\undec}{\text{undecidable}}
\renewcommand{\between}[2]{#1..#2}
\newcommand{\ours}[1]{\text{\hbox to 0pt{${}^{\text{\textit{#1}}}$\hfill}}}

\newcommand{\intto}{\,\text{..}\,}
\newcommand{\questeq}{=^?\!}
\newcommand{\questin}{\in^?\!}
\newcommand{\slpeq}{\equiv}
\newcommand{\cycshift}[2]{{#1} \curvearrowleft {#2}}

\newcommand{\intor}{\cup}
\newcommand{\intdoubled}{\times 2}
\newcommand{\intiter}[1]{{#1} {\times} \mathbb N}
\newcommand{\intexpop}{\{+, \intor, \intdoubled, \intiter{}\}}
\newcommand{\intsem}[1]{S(#1)}

\newcommand{\cproblem}[1]{\textsc{#1}}
\newcommand{\Empty}{\cproblem{UDPDA-Emptiness}}
\newcommand{\Univers}{\cproblem{UDPDA-Universality}}
\newcommand{\Member}{\cproblem{UDPDA-Compressed-Membership}}
\newcommand{\Equiv}{\cproblem{UDPDA-Equivalence}}
\newcommand{\Inclusion}{\cproblem{UDPDA-Inclusion}}
\newcommand{\CompSLP}[1]{\cproblem{SLP-Compo\-nent\-wise-${#1}$}}
\newcommand{\CompSLPbin}{\CompSLP{(0 \le 1)}}
\newcommand{\PartialMatch}{\cproblem{SLP-Partial-Word-Matching}}
\newcommand{\NondetUnivers}{\cproblem{Unary-PDA-Universality}}
\newcommand{\SubsetSum}{\cproblem{Subset-Sum}}
\newcommand{\GenSubsetSum}{\cproblem{Generalized-Subset-Sum}}
\newcommand{\QuerySLP}{\cproblem{SLP-Query}}
\newcommand{\EquivSLP}{\cproblem{SLP-Equivalence}}
\newcommand{\IntExpUnivers}{\cproblem{Integer-$\intexpop$-Expression-Univer\-sality}}

\newcommand{\Bin}{\{0, 1\}}

\newcommand{\str}[1]{\mathrm{eval}(#1)}
\newcommand{\Prog}{\mathcal P}
\newcommand{\Pref}{\mathcal P'}
\newcommand{\Loop}{\mathcal P''}
\newcommand{\Pair}{(\Pref, \Loop)}
\newcommand{\TransPref}{\mathcal T'}
\newcommand{\TransLoop}{\mathcal T''}
\newcommand{\TransPair}{(\TransPref, \TransLoop)}
\newcommand{\prodto}{\to}

\newcommand{\Qint}{Q_{0}}
\newcommand{\Qpop}{Q_{-1}}
\newcommand{\Qpush}{Q_{+1}}
\newcommand{\transint}{\delta_{0}}
\newcommand{\transpop}{\delta_{-1}}
\newcommand{\transpush}{\delta_{+1}}
\newcommand{\ints}[2]{\transint(#1)=#2}
\newcommand{\pops}[3]{\transpop(#1,#3)=#2}
\newcommand{\pushes}[3]{\transpush(#1)=(#2,#3)}
\newcommand{\Transitions}{\delta}

\newcommand{\NonRet}{\mathrm{NonRet}}
\newcommand{\fE}{\mathcal E}
\newcommand{\fH}{\mathcal H}
\newcommand{\fW}{\mathcal W}
\newcommand{\fG}{\mathcal G}
\newcommand{\fEbot}{\fE^\bot}
\newcommand{\Ebot}{E^\bot}

\DeclareMathOperator{\dom}{dom}

\newcommand{\fsym}{f}
\newcommand{\HistAlph}{\{a, \fsym\}}

\newcommand{\conv}{\cdot}
\newcommand{\qmark}{\text{\texttt{?}}}

\newcommand{\eps}{\varepsilon}
\newcommand{\sset}{\subseteq}

\newcommand{\mynote}[1]{}

\spnewtheorem{myclaim}{Claim}{\itshape}{\upshape}
\spnewtheorem*{remark*}{Remark}{\itshape}{\upshape}
\spnewtheorem*{pidea}{Proof idea}{\itshape}{\rmfamily}

\begin{document}

\title{Unary Pushdown Automata\\and Straight-Line Programs}

\author{Dmitry Chistikov \and Rupak Majumdar}
\institute{%
Max Planck Institute for Software Systems (MPI-SWS)\\
Kaiserslautern and Saarbr\"ucken, Germany\\
\email{\{dch,rupak\}@mpi-sws.org}
}

\authorrunning{\quad}
\titlerunning{\quad}

\maketitle

\begin{abstract}
We consider decision problems for deterministic pushdown automata over a unary alphabet
(udpda, for short).
Udpda are a simple computation model that accept exactly the unary regular languages,
but can be exponentially more succinct than finite-state automata.
We complete the complexity landscape for udpda by showing that emptiness (and thus universality)
is $\P$-hard, equivalence and compressed membership problems are $\P$-complete,
and inclusion is $\coNP$-complete.
Our upper bounds are based on a \emph{translation theorem} between udpda
and straight-line programs over the binary alphabet (SLPs).
We show that the characteristic sequence of any udpda can be represented as a 
pair of SLPs---one for the prefix, one for the lasso---that have size
linear in the size of the udpda and can be computed in polynomial time.
Hence, decision problems on udpda are reduced to decision problems on SLPs.
Conversely, any SLP can be converted in logarithmic space into a udpda, and this forms the
basis for our lower bound proofs.
We show $\coNP$-hardness of the ordered matching problem for SLPs, from which we derive $\coNP$-hardness
for inclusion.
In addition, we complete the complexity landscape for unary nondeterministic pushdown automata
by showing that the universality problem is $\PiTwoP$-hard, using a new class of integer expressions.
Our techniques have applications beyond udpda.
We show that our results imply $\PiTwoP$-completeness for a natural fragment of Presburger arithmetic
and $\coNP$ lower bounds for compressed matching problems with one-character wildcards.
\end{abstract}

\section{Introduction}
\label{s:intro}

Any model of computation comes with a set of fundamental decision
questions: emptiness (does a machine accept some input?),
universality (does it accept all inputs?), inclusion (are all inputs accepted by one
machine also accepted by another?),
and
equivalence (do two machines accept exactly the same inputs?).
The theoretical computer science community has
a fairly good understanding of the precise complexity of these problems for most
``classical'' models, such as finite and pushdown automata, with only a few
prominent open questions (e.\,g., the precise complexity of equivalence
for deterministic pushdown automata).

In this paper, we study a simple class of machines: deterministic
pushdown automata working on unary alphabets (unary dpda, or \emph{udpda}
for short).
A classic theorem of Ginsburg and Rice~\cite{ucflreg}
shows that they accept exactly the unary regular languages,
albeit with potentially exponential succinctness
when compared to finite automata.
However, the precise complexity of most basic decision
problems for udpda has remained open.

Our first and main contribution is that we close the complexity picture for these devices.
We show that emptiness is already $\P$-hard for udpda (even when the stack
is bounded by a linear function of the number of states) and thus
$\P$-complete.
By closure under complementation, it follows that universality is $\P$-complete as well.
Our main technical construction shows
equivalence is in $\P$ (and so $\P$-complete).
Somewhat unexpectedly, inclusion is $\coNP$-complete.
In addition, we study the \emph{compressed membership} problem:
given a udpda over the alphabet $\aalph$ and a number $n$ in binary, is $a^n$ in the language?
We show that this problem is $\P$-complete too.

A natural attempt at a decision procedure for equivalence or compressed membership would
go through translations to finite automata (since udpda only accept regular languages,
such a translation is possible).
Unfortunately, these automata can be exponentially larger than the udpda and, as we demonstrate,
such algorithms are not optimal.
Instead, our approach establishes a connection to straight-line programs (\emph{SLPs}) on binary words
(see, e.\,g., Lohrey~\cite{l12}).
An SLP $\Prog$ is a context-free grammar generating a single word, denoted $\str{\Prog}$, over $\Bin$.
Our main construction is a translation theorem: for any udpda, we construct in polynomial
time two SLPs $\Pref$ and $\Loop$ such that the infinite sequence $\str{\Pref} \cdot \str{\Loop}^\omega \in\Bin^\omega$
is the characteristic sequence of the language of the udpda
(for any $i \ge 0$, its $i$th element is $1$ iff $a^i$ is in the language).
With this construction, decision problems on udpda reduce to decision problems on compressed words.
Conversely, we show that from any pair $\Pair$ of SLPs, one can compute, in logarithmic space, a udpda accepting
the language with characteristic sequence $\str{\Pref}\cdot\str{\Loop}^\omega$.
Thus, as regards the computational complexity of decision problems,
lower bounds for udpda may be obtained from lower bounds for SLPs.
Indeed, we show $\coNP$-hardness of inclusion via $\coNP$-hardness of the {\em ordered matching}
problem for compressed words (i.\,e., is $\str{\Prog_1} \leq \str{\Prog_2}$ letter-by-letter, where the alphabet
comes with an ordering $\leq$), a problem of independent interest.

As a second contribution, we complete the complexity picture for unary
\emph{non-deterministic} pushdown automata (unpda, for short).
For unpda, the precise complexity of most decision problems was already known~\cite{huynh}.
The remaining open question was the precise complexity of the
universality problem, and we show that it is $\PiTwoP$-hard
(membership in $\PiTwoP$ was shown earlier by Huynh~\cite{huynh}).
An equivalent question was left open in Kopczy\'nski and To~\cite{kto10} in~2010,
but the question was posed as early as in~1976
by Hunt~III, Rosenkrantz, and Szymanski~\cite[open problem~2]{hrs}, where it was asked whether the
problem was in $\NP$ or $\PSPACE$ or outside both.
Huynh's $\PiTwoP$-completeness result for equivalence~\cite{huynh} showed, in
particular, that universality was in $\PSPACE$, and
our $\PiTwoP$-hardness result reveals that membership in $\NP$ is unlikely under usual
complexity assumptions.
As a corollary, we characterize the complexity
of the $\forall_{\mathrm{bounded}}\,\exists^*$-fragment of Presburger arithmetic,
where the universal quantifier ranges over numbers at most exponential in the size of the formula.

To show $\PiTwoP$-hardness, we show hardness of the universality problem for a class
of integer expressions. Several decision problems of this form, with the set of operations $\{+, \intor\}$,
were studied in the classic paper of Stockmeyer and Meyer~\cite{sm73},
and we show that checking universality of expressions over $\intexpop$
is $\PiTwoP$-complete (the upper bound follows from Huynh~\cite{huynh}).

\smallskip\noindent
\textbf{Related work.}
Table~\ref{table} provides the current complexity picture, including
the results in this paper.
Results on general alphabets are mostly classical and included for comparison.
Note that the complexity landscape for udpda differs from those for unpda, dpda, and finite automata.
Upper bounds for emptiness and universality are classical, and
the lower bounds for emptiness are originally by Jones and Laaser~\cite{jl76}
and Goldschlager~\cite{gol81}.
In the nondeterministic unary case,
$\NP$-completeness of compressed membership
is from Huynh~\cite{huynh}, rediscovered later by
Plandowski and Rytter~\cite{prjewels}.
The $\PSPACE$-completeness of the compressed membership problem for
binary pushdown automata (see definition in Section~\ref{s:disc}) is by Lohrey~\cite{l06}.

The main remaining open question is the precise complexity of the equivalence
problem for dpda. It was shown decidable by S\'{e}nizergues~\cite{sen} and primitive
recursive by Stirling~\cite{stir} and Jan\v{c}ar~\cite{jan}, but only $\P$-hardness (from emptiness)
is currently known.
Recently, the equivalence question for dpda when the stack alphabet
is unary was shown to be $\NL$-complete by B\"ohm, G\"oller, and Jan\v{c}ar~\cite{bgj13}.
From this, it is easy to show that emptiness and universality are also $\NL$-complete.
Compressed membership, however, remains $\PSPACE$-complete
(see Caussinus et~al.~\cite{cmtv} and Lohrey~\cite{l11}),
and inclusion is, of course, already undecidable.
When we additionally restrict dpda to both unary input and unary stack alphabet,
all five decision problems are $\L$-complete.\mynote{Not in Appendix.}

We discuss corollaries of our results and other related work in Section~\ref{s:disc}.
While udpda are a simple class of machines,
our proofs show that reasoning about these machines can be
surprisingly subtle.


%

\smallskip
\noindent
\textbf{Acknowledgements.}
We thank Joshua Dunfield for discussions.

\begin{table}[t]
\caption{Complexity of decision problems for pushdown automata.}
\label{table}
\vspace*{-4ex}
\begin{equation*}
\arraycolsep=0.8em
\begin{array}{lcccc}
\hline\\[-2ex]
& \multicolumn{2}{c}{\text{unary}}
& \multicolumn{2}{c}{\text{binary}}
\\
& \text{dpda} & \text{npda}
& \text{dpda} & \text{npda}
\\[0.5ex]
\hline
\rule{0pt}{2.5ex}
\text{Emptiness} & \P\ours{l} & \P & \P & \P \\
\text{Universality} & \P\ours{l} & \PiTwoP\ours{l} & \P & \undec \\
\text{Equivalence} & \P\ours{u,l} & \PiTwoP & \between \P \text{\,pr.rec.} & \undec \\
\text{Inclusion} & \coNP\ours{u,l} & \PiTwoP & \undec & \undec \\
\text{Compressed membership} & \P\ours{u,l} & \NP & \PSPACE & \PSPACE \\
\hline
\end{array}
\end{equation*}
\begin{quote}

Legend:
``dpda'' and ``npda'' stand for deterministic and \emph{possibly}
nondeterministic pushdown automata, respectively;
``unary'' and ``binary'' refer to their input alphabets.
Names of complexity classes stand for completeness
with respect to logarithmic-space reductions;
abbreviation ``pr.rec.'' stands for ``primitive recursive''.
Superscripts \textit{u} and \textit{l} denote new upper and lower bounds
shown in this paper.
\end{quote}
\end{table}

\section{Preliminaries}
\label{s:pre}

\subsubsection*{Pushdown automata.}
%
A \df{unary pushdown automaton} (unpda) over the alphabet $\aalph$ is a finite structure
$\A = (Q, \StackAlph, \bot, q_0, F, \Transitions)$, with $Q$ a set of (control) states,
$\StackAlph$ a stack alphabet, $\bot \in \StackAlph$ a bottom-of-the-stack symbol,
$q_0 \in Q$ an initial state, $F \sset Q$ a set of final states, and
$\Transitions \sset
 (Q \times (\aalph \cup \{\eps\}) \times \StackAlph) \times
 (Q \times \StackAlph^*)$
a set of transitions with the property that,
for every $(q_1, \sigma, \gamma, q_2, s) \in \Transitions$,
either $\gamma \ne \bot$ and $s \in (\StackAlph \setminus \{\bot\})^*$,
or $\gamma = \bot$ and $s \in \{\eps\} \cup (\StackAlph \setminus \{\bot\})^* \bot$.
Here and everywhere below $\eps$ denotes the empty word.

The semantics of unpda is defined in the following standard way.
The set of \df{configurations} of $\A$ is $Q \times (\StackAlph \setminus \{\bot\})^* \bot$.
Suppose $(q_1, s_1)$ and $(q_2, s_2)$ are configurations;
we write $(q_1, s_1) \confto_\sigma \! (q_2, s_2)$
and say that a \df{move} to $(q_2, s_2)$
is available to $\A$ at $(q_1, s_1)$ iff
there exists a transition $(q_1, \sigma, \gamma, q_2, s) \in \Transitions$
such that, if $\gamma \ne \bot$ or $s \ne \eps$, then $s_1 = \gamma s'$
and $s_2 = s s'$ for some $s' \in \StackAlph^*$,
or, if $\gamma = \bot$ and $s = \eps$, then $s_1 = s_2 = \bot$.%
\mynote{Here, for $\gamma = \bot$, there is no difference between internal
    and pop transitions. It only appears in the proof of Theorem~\ref{th:ip},
    when we introduce our normal form in Subsection~\ref{ss:proof:pre}.}
A~unary pushdown automaton is called \df{deterministic}, shortened to \df{udpda},
if at every configuration at most one move is available.


A word $w \in \aalph^*$ is \df{accepted} by $\A$ if there exists a configuration $(q_k, s_k)$
with $q_k \in F$ and a sequence of moves
$(q_i, s_i) \confto_{\sigma_i} \! (q_{i + 1}, s_{i + 1})$, $i = 0, \ldots, k - 1$,
such that $s_0 = \bot$ and $\sigma_0 \ldots \sigma_{k - 1} = w$;
that is, the acceptance is by final state.
The \df{language} of $\A$, denoted $L(\A)$, is the set of all words $w \in \aalph^*$
accepted by $\A$.

We define the \df{size} of a unary pushdown automaton $\A$
as $|Q| \cdot |\StackAlph|$, provided that
for all transitions $(q_1, \sigma, \gamma, q_2, s) \in \Transitions$
the length of the word~$s$ is at most~$2$
(see also~\cite{pudpda}).
While this definition is better suited for deterministic
rather than nondeterministic automata, it already suffices for the purposes
of Section~\ref{s:nondet}, where we handle unpda,
because it is always the case that $|\Transitions| \le 2 \, |Q|^2 \, |\StackAlph|^4$.

\vspace{-3ex}
\subsubsection*{Decision problems.}
%
We consider the following decision problems: emptiness ($L(\A) \questeq \emptyset$),
universality ($L(\A) \questeq \aalph^*$), equivalence ($L(\A_1) \questeq L(\A_2)$),
and inclusion ($L(\A_1) \sset^?\! L(\A_2)$). The \df{compressed membership} problem
for unary pushdown automata is associated with the question $a^n \questin L(\A)$,
with $n$ given in binary as part of the input.
In the following, hardness is with respect to logarithmic-space reductions.
%
Our first result is that emptiness is already $\P$-hard for udpda.

\begin{proposition}
\label{p:al}
\Empty\ and \Univers\ are $\P$-complete.
\end{proposition}

\begin{proof}
Emptiness is in $\P$ for non-deterministic pushdown automata on any
alphabet, and deterministic automata can be complemented in polynomial
time.
So, we focus on showing $\P$-hardness for emptiness.

We encode the computations of an alternating logspace Turing
machine using an udpda.
We assume without loss of generality that each configuration $c$ of the machine has exactly two
successors, denoted $c_l$ (left successor) and $c_r$ (right
successor), and that each run of the machine terminates.
The udpda encodes a successful run over the computation tree of the TM.
The states of the udpda are configurations of the Turing machine and
an additional \emph{context}, which can be $a$ (``accepting''),
$r$ (``rejecting''), or $x$ (``exploring'').
The stack alphabet consists of pairs $(c,d)$, where $c$ is a
configurations of the machine and the direction
$d\in\set{l,r}$, together with
an additional end-of-stack symbol.
The alphabet has just one symbol $0$.
The initial state is $(c_0, x)$, where $c_0$ is the initial
configuration of the machine,
and the stack has the end-of-stack symbol.

Suppose the current state is $(c,x)$, where $c$ is not an accepting or
rejecting configuration.
The udpda pushes $(c,l)$ on to the stack and updates its state to
$(c_l, x)$, where $c_l$ is the left successor of $c$.
The invariant is that in the exploring phase, the stack maintains the current path in the
computation tree, and if the top of the stack is $(c,l)$ (resp.\ $(c,r)$) then the
current state is the left (resp.\ right) successor of $c$.

Suppose now the current state is $(c,x)$ where $c$ is an accepting
configuration.
The context is set to $a$, and the udpda backtracks up the computation
tree using the stack.
If the top of the stack is the end-of-stack symbol, the machine accepts.
If the top of the stack $(c', d)$ consists of an existential
configuration $c'$, then the new state is $(c',a)$ and recursively the
machine moves up the stack.
If the top of the stack $(c',d)$ consists of a universal configuration
$c'$, and $d=l$, then the new state is $(c'_r, x)$, the right
successor of $c'$, and the top of
stack is replaced by $(c',r)$.
If the top of the stack $(c',d)$ consists of a universal configuration
$c'$, and $d=r$, then the new state is $(c', a)$, the stack is popped,
and the machine continues to move up the stack.

Suppose now the current state is $(c,x)$ where $c$ is a rejecting
configuration.
The context is set to $r$, and the udpda backtracks up the computation
tree using the stack.
If the top of the stack is the end-of-stack symbol, the machine rejects.
If the top of the stack $(c', d)$ consists of an existential
configuration $c'$, and $d=l$, then the new state is $(c'_r,x)$ and
the top of stack is replaced with $(c',r)$.
If the top of the stack $(c', d)$ consists of an existential
configuration $c'$, and $d=r$, then the new state is $(c',r)$, the
stack
is popped, and the machine continues to move up the stack.
The top of stack is replaced with $(c',r)$.
If the top of the stack $(c',d)$ consists of a universal configuration
$c'$, then the new state is $(c',r)$, the stack is popped,
and the machine continues to move up the stack.

It is easy to see that the reduction can be performed in logarithmic
space.
If the TM has an accepting computation tree, the udpda has an
accepting run following the computation tree.
On the other hand, any accepting computation of the udpda is a
depth-first traversal of an accepting computation tree of the TM.

Finally, since udpda can be complemented in logarithmic space, we get
the corresponding results for universality.
This completes the proof.
\qed
\end{proof}

\vspace{-3ex}
\subsubsection*{Straight-line programs.}
%
%
A \df{straight-line program}~\cite{l12}, or an \df{SLP}, over an alphabet $\Terminals$ is a context-free grammar
that generates a single word;
in other words, it is a tuple $\Prog = (S, \Terminals, \Nonterminals, \Productions)$,
where $\Terminals$ and $\Nonterminals$ are disjoint sets of
\df{terminal} and \df{nonterminal} symbols (\df{terminals} and \df{nonterminals}),
$S \in \Nonterminals$ is the \df{axiom},
and the function $\Productions \colon \Nonterminals \to (\Terminals \cup \Nonterminals)^*$
defines a set of \df{productions} written as ``$N \prodto w$'', $w = \Productions(N)$,
and satisfies the property that the relation $\{ (N, D) \mid N \prodto w \text{\ and\ }
D \text{\ occurs in\ } w \}$ is acyclic.
An SLP~$\Prog$ is said to \df{generate} a (unique) word $w \in \Terminals^*$, denoted $\str \Prog$,
which is the result of applying substitutions $\pi$ to $S$.


An SLP is said to be in \df{Chomsky normal form} if for all productions $N \prodto w$
it holds that either $w \in \Terminals$ or $w \in \Nonterminals^{\!2}$.
The \df{size} of an SLP is the number of nonterminals in its Chomsky normal form.%
\mynote{Is this rigorous enough? we do not specify how to get to CNF, that is,
 what is \emph{the} CNF of an SLP.}

%


\section{Indicator pairs and the translation theorem}
\label{s:ip}

We say that a pair of SLPs $\Pair$ over an alphabet $\Terminals$
\df{generates} a sequence $c \in \Terminals^\omega$ if $\str\Pref \cdot (\str\Loop)^\omega = c$.
We call an infinite sequence $c \in \{0, 1\}^\omega$, $c = c_0 c_1 c_2 \ldots$\,,
the \df{characteristic sequence} of a unary language $L \sset \aalph^*$
if, for all $i \ge 0$, it holds that $c_i$ is $1$ if $a^i \in L$ and $0$ otherwise.
One may note that the characteristic sequence is eventually periodic
if and only if $L$ is regular.

\begin{definition}
\label{def:ip}
A pair of straight-line programs $\Pair$ over $\Bin$ is called
an \df{indicator pair} for a unary language $L \sset \aalph^*$ if
it generates the characteristic sequence of $L$.
\end{definition}

\noindent
A unary language can have several different indicator pairs.
Indicator pairs form a descriptional system for unary languages, with
the \df{size} of $\Pair$ defined as the sum of sizes of $\Pref$ and $\Loop$.
The following translation theorem shows that udpda and indicator pairs are
polynomially-equivalent representations for unary regular languages.
%
%
We remark that Theorem~\ref{th:ip} does not give
a normal form for udpda because of the non-uniqueness of indicator pairs.

\begin{theorem}[translation theorem]
\label{th:ip}
For a unary language $L \sset \aalph^*$\textup{:}
\vspace*{-1.5ex}
\begin{enumerate}
\renewcommand{\theenumi}{\textup{\arabic{enumi}}}
\renewcommand{\labelenumi}{\textup{(}\theenumi\textup{)}}
\item\label{th:ip:pdatoip}
      if there exists a udpda $\A$ of size $m$ with $L(\A) = L$, then
      there exists an indicator pair for $L$ of size $O(m)$;
\item\label{th:ip:iptopda}
      if there exists an indicator pair for $L$ of size $m$, then
      there exists a udpda $\A$ of size $O(m)$ with $L(\A) = L$.
\end{enumerate}
\vspace*{-1.5ex}
Both statements are supported by polynomial-time algorithms,
the second of which works in logarithmic space.
\end{theorem}

\begin{pidea}
We only discuss part~\ref{th:ip:pdatoip}, which presents the main
technical challenge.
The starting point is the simple observation that a udpda $\A$
has a single infinite computation, provided that the input tape
supplies $\A$ with as many input symbols $a$ as it needs to consume.
Along this computation, \emph{events} of two types are encountered:
$\A$ can consume a symbol from the input and can enter a final state.

The crucial technical task is to construct inductively, using dynamic
programming, straight-line programs that record these events along finite
computational segments. These segments are of two types: first,
between matching push and pop moves (``procedure calls'') and, second,
from some starting point until a move pops the symbol that has been on top
of the stack at that point (``exits from current context''). Loops are
detected, and infinite computations are associated with pairs of SLPs:
in such a pair, one SLP records the initial segment, or prefix of the
computation, and the other SLP records events within the loop.

After constructing these SLPs, it remains to transform the computational
``history'', or \emph{transcript}, associated with the initial
configuration of $\A$ into the characteristic sequence. This
transformation can easily be performed in polynomial time, without
expanding SLPs into the words that they generate. The result is the
indicator pair for $\A$.
\qed
\end{pidea}

\noindent
The full proof of Theorem~\ref{th:ip} is given is Section~\ref{s:proof}.
Note that going from indicator pairs to udpda is useful for obtaining
lower bounds on the computational complexity of decision problems for udpda
(Theorems~\ref{th:member} and~\ref{th:incl}).
For this purpose, it suffices to model just
a single SLP, but taking into account the whole pair is interesting from the
point of view of descriptional complexity (see also Section~\ref{s:disc}).

\section{Decision problems for udpda}
\label{s:app}

\subsection{Compressed membership and equivalence}
\label{ss:meq}


For an SLP $\Prog$, by $|\Prog|$ we denote the length of the word $\str\Prog$, and
by $\Prog[n]$ the $n$th symbol of $\str\Prog$, counting from $0$
(that is, $0 \le n \le |\Prog| - 1$).
We write $\Prog_1 \slpeq \Prog_2$ if and only if $\str{\Prog_1} = \str{\Prog_2}$.

The following \QuerySLP\ problem is known to be $\P$-complete
(see Lifshits and Lohrey~\cite{ll06}): given an SLP $\Prog$ over $\Bin$ and
a number $n$ in binary, decide whether $\Prog[n] = 1$.
The problem \EquivSLP\ is only known to be in~$\P$
(see, e.\,g., Lohrey~\cite{l12}):
given two SLPs $\Prog_1$, $\Prog_2$,
decide whether $\Prog_1 \slpeq \Prog_2$.

\begin{theorem}
\label{th:member}
\Member\ is $\P$-complete.
\end{theorem}

\begin{proof}
The upper bound follows from Theorem~\ref{th:ip}.
Indeed, given a udpda $\A$ and a number $n$,
first construct an indicator pair $\Pair$ for $L(\A)$. Now
compute $|\Pref|$ and $|\Loop|$ and then decide if $n \le |\Pref| - 1$.
If so, the answer is given by $\Pref[n]$, otherwise by $\Loop[r]$,
where $r = (n - |\Pref|) \bmod |\Loop|$ and in both cases $1$ is interpreted
as ``yes'' and $0$ as ``no''.
\par
To prove the lower bound, we reduce from
the \QuerySLP\ problem. Take an instance with an SLP $\Prog$
and a number $n$ in binary.
By transforming the pair $(\Prog, \Prog_0)$, with $\Prog_0$ any fixed SLP over $\Bin$,
into a udpda $\A$ using part~\ref{th:ip:iptopda} of Theorem~\ref{th:ip},
this problem is reduced, in logspace, to whether $a^n \in L(\A)$.
\qed
\end{proof}

\noindent
Recall that emptiness and universality of udpda are $\P$-complete by Proposition~\ref{p:al}.
Our next theorem extends this result to the general equivalence problem for udpda.

\begin{theorem}
\label{th:equiv}
\Equiv\ is $\P$-complete.
\end{theorem}

\begin{proof}
Hardness follows from Proposition~\ref{p:al}. We show how Theorem~\ref{th:ip}
can be used to prove the upper bound: given udpda $\A_1$ and $\A_2$, first construct indicator pairs $(\Pref_1, \Loop_1)$
and $(\Pref_2, \Loop_2)$ for $L(\A_1)$ and $L(\A_2)$, respectively. Now reduce
the problem of whether $L(\A_1) = L(\A_2)$ to \EquivSLP.
The key observation is that an eventually periodic sequence that has periods $|\Loop_1|$ and $|\Loop_2|$
also has period $t = \gcd(|\Loop_1|, |\Loop_2|)$. Therefore, it suffices to check
that, first, the initial segments of the generated sequences match and, second,
that $\Loop_1$ and $\Loop_2$ generate powers of the same word up to a certain circular shift.
\par
In more detail,
let us first introduce some auxiliary operations for SLP.
For SLPs $\Prog_1$ and $\Prog_2$, by $\Prog_1 \cdot \Prog_2$ we denote
an SLP that generates $\str{\Prog_1} \cdot \str{\Prog_2}$, obtained by
``concatenating'' $\Prog_1$ and $\Prog_2$.
Now suppose that $\Prog$ generates $w = w[0] \ldots w[|\Prog|-1]$.
Then the SLP $\Prog[a \intto b)$ generates the word $w[a \intto b) = w[a] \ldots w[b - 1]$,
of length $b - a$ (as in $\Prog[n]$, indexing starts from $0$).
The SLP $\Prog^\alpha$ generates $w^\alpha$, with the meaning clear for $\alpha = 0, 1, 2, \ldots$\,,
also extended to $\alpha \in \mathbb Q$ with $\alpha \cdot |\Prog| \in \mathbb Z_{\ge 0}$
by setting $w^{k + n / |w|} = w^k \cdot w[0 \intto n)$, $n < |w|$.
Finally, $\cycshift{\Prog}{s}$ denotes \df{cyclic shift} and evaluates
to $w[s \intto |w|) \cdot w[0 \intto s)$.
One can easily demonstrate that all these operations can be implemented
in polynomial time.
\par
So, assume that $|\Pref_1| \ge |\Pref_2|$.
First, one needs to check whether $\Pref_1 \slpeq \Pref_2 \cdot (\Loop_2)^\alpha$,
where $\alpha = (|\Pref_1| - |\Pref_2|) / |\Loop_2|$.
Second, note that an eventually periodic sequence that has periods $|\Loop_1|$ and $|\Loop_2|$
also has period $t = \gcd(|\Loop_1|, |\Loop_2|)$. Compute $t$ and an auxiliary
SLP $\Loop = \Loop_1[0 \intto t)$, and then check whether
$\Loop_1 \slpeq (\Loop)^{|\Loop_1| / t}$ and
$\cycshift{\Loop_2}{s} \slpeq (\Loop)^{|\Loop_2| / t}$ with $s = (|\Pref_1| - |\Pref_2|) \bmod |\Loop_2|$.
It is an easy exercise to show that $L(\A_1) = L(\A_2)$ iff all the checks are
successful. This completes the proof.
\qed
\end{proof}

\subsection{Inclusion}
\label{ss:incl}

A natural idea for handling the inclusion problem for udpda would be to extend the result
of Theorem~\ref{th:equiv}, that is, to tackle inclusion similarly to equivalence.
This raises the problem of comparing the words generated by two SLPs
in the componentwise sense with respect to the order $0 \le 1$. To the best of
our knowledge, this problem has not been studied previously, so we deal with it separately.
As it turns out, here one cannot hope for an efficient algorithm unless $\P = \NP$.

Let us define the following family of problems, parameterized by
partial order $R$ on the alphabet of size at least $2$, and denoted $\CompSLP{R}$. The input
is a pair of SLPs $\Prog_1$, $\Prog_2$ over an alphabet partially ordered by $R$, 
generating words of equal length. 
The output is ``yes'' iff for all $i$, $0 \le i < |\Prog_1|$,
the relation $R(\Prog_1[i], \Prog_2[i])$ holds.
By \CompSLPbin\ we mean a special case of this problem where $R$ is the
partial order on $\Bin$ given by $0 \le 0$, $0 \le 1$, $1 \le 1$.

\begin{theorem}
\label{th:comporder}
$\CompSLP{R}$ is $\coNP$-complete if $R$ is not the equality relation (that is, if $R(a, b)$ holds
for some $a \ne b$), and in $\P$\! otherwise.
\end{theorem}

\begin{proof}
We first prove that \CompSLPbin\ is $\coNP$-hard.
We show a reduction from the complement of \SubsetSum:
suppose we start with an instance of \SubsetSum\ containing a vector of naturals
$w = (w_1, \ldots, w_n)$ and a natural $t$, and the question
is whether there exists a vector $x = (x_1, \ldots, x_n) \in \Bin^n$
such that $x \conv w = t$, where $x \conv w$ is defined as
the inner product $\sum_{i = 1}^{n} x_i w_i$. Let $s = (1, \ldots, 1) \conv w$ be
the sum of all components of $w$.

We use the construction of so-called Lohrey words.
Lohrey shows~\cite[Theorem~5.2]{l06} that given such an instance, it is possible
to construct in logarithmic space two SLPs that generate words
$W_1 = \prod_{x \in \Bin^n} a^{x \conv w} b a^{s - x \conv w}$ and
$W_2 = (a^{t} b a^{s - t})^{2^n}$, where the product in $W_1$
enumerates the $x$es in the lexicographic order. Now $W_1$ and $W_2$ share a symbol $b$
in some position iff the original instance of \SubsetSum\ is a yes-instance.
Substitute $0$ for $a$ and $1$ for $b$ in the first SLP,
and $0$ for $b$ and $1$ for $a$ in the second SLP.
The new SLPs $\Prog_1$ and $\Prog_2$ obtained in this way
form a no-instance of \CompSLPbin\ iff the original instance
of \SubsetSum\ is a yes-instance, because now the ``distinguished'' pair of
symbols consists of a $1$ in $\Prog_1$ and $0$ in $\Prog_2$.
Therefore, \CompSLPbin\ is $\coNP$-hard.

Now observe that, for any $R$, membership of $\CompSLP{R}$
in $\coNP$ is obvious, and the hardness
is by a simple reduction from \CompSLPbin: just substitute $a$ for $0$
and $b$ for $1$ (recall that by the definition of partial order,
$R(b, a)$ would entail $a = b$, which is false).
In the last special case in the statement, $R$ is just the equality relation,
so deciding $\CompSLP{R}$ is the same as deciding \EquivSLP,
which is in~$\P$ (see Section~\ref{s:app}).
This concludes the proof.
\qed
\end{proof}

\noindent
A corollary of Theorem~\ref{th:comporder} on a problem of matching for
compressed partial words is demonstrated in Section~\ref{s:disc}.

\begin{remark*}
An alternative reduction showing hardness of \CompSLPbin, this time from $\cproblem{Circuit-SAT}$,
but also making use of \SubsetSum\ and Lohrey words,
can be derived from Bertoni, Choffrut, and Radicioni~\cite[Lemma~3]{bcr}.
They show that
for any Boolean circuit with NAND-gates there exists a pair of straight-line programs $\Prog_1$,
$\Prog_2$ generating words over $\Bin$ of the same length with the following property:
the function computed by the circuit takes on the value $1$ on at least one
input combination iff both words share a $1$ at some position.
Moreover, these two SLPs can be constructed in polynomial time.
As a result, after flipping all terminal symbols in the second of these SLPs,
the resulting pair is a no-instance of \CompSLPbin\ iff
the original circuit is satisfiable.
\end{remark*}

\begin{theorem}
\label{th:incl}
\Inclusion\ is $\coNP$-complete.
\end{theorem}

\begin{proof}
First combine Theorem~\ref{th:comporder} with part~\ref{th:ip:iptopda}
of Theorem~\ref{th:ip} to prove hardness.
Indeed, Theorem~\ref{th:comporder} shows that \CompSLPbin\ is $\coNP$-hard.
Take an instance with two SLPs $\Prog_1$, $\Prog_2$ and transform indicator pairs
$(\Prog_1, \Prog_0)$ and $(\Prog_2, \Prog_0)$, with $\Prog_0$ any fixed SLP over $\Bin$,
into udpda $\A_1$, $\A_2$ with the help of part~\ref{th:ip:iptopda} of Theorem~\ref{th:ip}.
Now the characteristic sequence of $L(\A_i)$, $i = 1, 2$, is equal
to $\str{\Prog_i} \cdot (\str{\Prog_0})^\omega$. As a result, it holds that
$\str{\Prog_1} \le \str{\Prog_2}$ in the componentwise sense
if and only if $L(\A_1) \sset L(\A_2)$.
This concludes the hardness proof.

It remains to show that \Inclusion\ is in $\coNP$.
First note that for any udpda $\A$ there exists a deterministic pushdown automaton (DFA)
that accepts $L(\A)$ and has size at most $2^{O(m)}$, where $m$ is the size of $\A$
(see discussion in Section~\ref{s:disc} or Pighizzini~\cite[Theorem~8]{pudpda}).
Therefore, if $L(\A_1) \not\sset L(\A_2)$, then there exists
a witness $a^n \in L(\A_2) \setminus L(\A_1)$ with $n$ at most exponential
in the size of $\A_1$ and $\A_2$.
By Theorem~\ref{th:member}, compressed membership is in $\P$, so
this completes the proof.
\qed
\end{proof}

\section{Proof of Theorem~\ref{th:ip}}
\label{s:proof}

Let us first recall some standard definitions and fix notation.
In a udpda $\A$, if $(q_1, s_1) \confto_\sigma \! (q_2, s_2)$ for some $\sigma$, we also write
$(q_1, s_1) \confto (q_2, s_2)$.
A \df{computation} of a udpda $\A$ starting at a configuration $(q, s)$
is defined as a (finite or infinite) sequence of configurations $(q_i, s_i)$
with $(q_1, s_1) = (q, s)$ and, for all $i$,
$(q_i, s_i) \confto_{\sigma_i} \! (q_{i + 1}, s_{i + 1})$
for some $\sigma_i$. If the sequence is finite and ends with $(q_k, s_k)$,
we also write $(q_1, s_1) \confto^*_w \! (q_k, s_k)$, where
$w = \sigma_1 \ldots \sigma_{k - 1} \in \aalph^*$.
We can also omit the word $w$ when it is not important
and say that $(q_k, s_k)$ is \df{reachable} from $(q_1, s_1)$;
in other words, the \df{reachability} relation $\confto^*$ is the reflexive
and transitive closure of the move relation $\confto$.

\subsection{From indicator pairs to udpda}
\label{ss:slptopda}

Going from indicator pairs to udpda is the easier direction in Theorem~\ref{th:ip}.
We start with an auxiliary lemma that enables one to model a single SLP with a udpda.
This lemma on its own is already sufficient for lower bounds
of Theorem~\ref{th:member} and Theorem~\ref{th:incl} in Section~\ref{s:app}.

\begin{lemma}
\label{l:slptopda}
There exists an algorithm that works in logarithmic space and transforms
an arbitrary SLP $\Prog$ of size $m$ over $\Bin$
into a udpda $\A$ of size $O(m)$ over $\aalph$ such that
the characteristic sequence of $L(\A)$ is $0 \cdot \str\Prog \cdot 0^\omega$.
In $\A$, it holds that $(q_0, \bot) \confto^*_w \! (\bar q_0, \bot)$
for $w = a^{|\str{\Prog}|}$, $q_0$ the initial state,
and $\bar q_0$ a non-final state without outgoing transitions.
\end{lemma}

\begin{proof}
The main part of the algorithm works as follows\mynote{easier with equations?}.
Assume that $\Prog$
is given in Chomsky normal form. With each nonterminal $N$ we associate
a gadget in the udpda $\A$, whose \df{interface} is by definition
the entry state $q_N$ and the exit state $\bar q_N$, which will only have outgoing pop transitions.
With a production of the form $N \prodto \sigma$, $\sigma \in \Bin$,
we associate a single internal transition from $q_N$ to $\bar q_N$ reading
an $a$ from the input tape. The state $q_N$ is always non-final, and
the state $\bar q_N$ is final if and only if $\sigma = 1$. With a production
of the form $N \prodto A B$ we associate two stack symbols $\gamma_N^1$, $\gamma_N^2$
and the following gadget.
At a state $q_N$, the udpda pushes a symbol $\gamma_N^1$ onto the stack
and goes to the state $q_A$. We add
a pop transition from $\bar q_A$ that reads $\gamma_N^1$ from the stack and leads
to an auxiliary state $q'_N$. The only transition from this state pushes
$\gamma_N^2$ and leads to $q_B$, and another transition from $\bar q_B$
pops $\gamma_N^2$ and goes to $\bar q_N$. Here all three states $q_N$, $q'_N$,
and $\bar q_N$ are non-final, and the four introduced incident transitions do not read
from the input. Finally, if a nonterminal $N$ is the axiom of $\Prog$,
make the state $q_N$ initial and non-final and make $\bar q_N$ a non-accepting sink
that reads~$a$ from the input and pops $\bot$.
The reader can easily check that the
characteristic sequence of the udpda $\A$ constructed in this way is
indeed $0 \cdot \str\Prog \cdot 0^\omega$,
and the construction can be performed in logarithmic space.
\par
Now note that while the udpda $\A$ satisfies $|Q| = O(m)$, we may have also introduced
up to $2$ stack symbols per nonterminal. Therefore, the size of $\A$ can be as large
as $\MyOmega(m^2)$. However, we can use a standard trick from circuit complexity
to avoid this blowup and make this size linear in $m$. Indeed, first observe
that the number of stack symbols, not counting $\bot$, in the construction
above can be reduced to~$k$, the maximum, over all nonterminals $N$, of the
number of occurrences of $N$ in the right-hand sides of productions of $\Prog$.
Second, recall that a straight-line program naturally defines a circuit where
productions of the form $N \prodto A B$ correspond to gates performing
concatenation. The value of $k$ is the maximum fan-out of gates in this
circuit, and it is well-known how to reduce it to $O(1)$ with just
a constant-factor increase in the number of gates (see, e.\,g.,
Savage~\cite[Theorem~9.2.1]{savage}).
The construction can be easily
performed in logarithmic space, and the only building block is the identity
gate, which in our case translates to a production of the form $N \prodto A$.
Although such productions are not allowed in Chomsky normal form, the
construction above can be adjusted accordingly, in a straightforward fashion.
This completes the proof.
\qed
\end{proof}

\noindent
Now, to model an entire indicator pair, we apply Lemma~\ref{l:slptopda} twice
and combine the results.

\begin{lemma}
\label{l:iptopda}
There exists an algorithm that works in logarithmic space and,
given an indicator pair $\Pair$ of size $m$ for some unary language $L \sset \aalph^*$,
outputs a udpda $\A$ of size $O(m)$ such that $L(\A) = L$.
\end{lemma}

\begin{proof}
We shall use the same notation as in Subsection~\ref{ss:meq} of Section~\ref{s:app}.
First compute the bit $b = \Pref[0]$ and construct an SLP $\Pref_1$
of size $O(m)$ such that $\str\Pref = b \cdot \str{\Pref_1}$.
Note that this can be done in logarithmic space, even though
the general \QuerySLP\ problem is $\P$-complete.
Now construct, according to Lemma~\ref{l:slptopda},
two udpda $\A'$ and $\A''$ for $\Pref_1$ and $\Loop$, respectively.
Assume that their sets of control states are disjoint and connect them
in the following way. Add internal $\eps$-transitions from the ``last''
states of both to the initial state of $\A''$. Now make the initial state
of $\A'$ the initial state of $\A$; make it also a final state if $b = 1$.
It is easily checked that the language of the udpda $\A$ constructed in this way
has characteristic sequence $\str\Pref \cdot (\str\Loop)^\omega$
and, hence, is equal to $L$.
\qed
\end{proof}

\noindent
Lemma~\ref{l:iptopda} proves part~\ref{th:ip:iptopda} in Theorem~\ref{th:ip}.

\subsection{From udpda to indicator pairs}
\label{ss:proof:pre}

Going from udpda to indicator pairs is the main part of Theorem~\ref{th:ip},
and in this subsection we describe our construction in detail.
The proof of the key technical lemma is deferred until the following
Subsection~\ref{ss:proof:translemma}.

\vspace{-3ex}
\subsubsection*{Assumptions and notation.}
We assume without loss of generality that the given udpda $\A$ satisfies the following
conditions. First, its set of control states, $Q$, is partitioned into three
subsets according to the type of available moves. More precisely, we assume%
\footnote{Here and further in the text
we use the symbol $\sqcup$ to denote the union of disjoint sets.}
$Q = \Qint \sqcup \Qpush \sqcup \Qpop$ with the property that all transitions
$(q, \sigma, \gamma, q', s)$
with states $q$ from $Q_d$, $d \in \{0, -1, +1\}$, have $|s| = 1 + d$;
moreover, we assume that $s = \gamma$ whenever $d = 0$,
and $s = \gamma' \gamma$ for some $\gamma' \in \StackAlph$ whenever $d = 1$.

Second, for convenience of notation we assume
that there exists a subset $R \sset Q$ such that all transitions departing from
states from $R$ read a symbol from the input tape, and transitions departing
from $Q \setminus R$ do not.

Third, we assume that $\Transitions$ is specified by means of total functions
$\transint \colon \Qint \to Q$,
$\transpush \colon \Qpush \to Q \times \StackAlph$, and
$\transpop \colon \Qpop \times \StackAlph \to Q$.
We write
$\ints{q}{q'}$, $\pushes{q}{q'}{\gamma}$, and $\pops{q}{q'}{\gamma}$ accordingly;
associated transitions and states are called \df{internal}, \df{push}, and \df{pop} transitions and states,
respectively.
Note that this assumption implies that only pop transitions can ``look'' at the top
of the stack.

\begin{myclaim}
\label{c:udpdanf}
An arbitrary udpda $\A' = (Q', \StackAlph, \bot, q'_0, F', \Transitions')$ of size $m$
can be transformed into a udpda $\A = (Q, \StackAlph, \bot, q_0, F, \Transitions)$ that
accepts $L(\A')$,
satisfies the assumptions of this subsubsection\mynote{FIXME}, and
has $|Q| = O(m)$ control states.
\end{myclaim}

\noindent
The proof is easy and left to the reader.

Note that since $\A$ is deterministic, it holds that for any configuration $(q, s)$ of $\A$
there exists a unique infinite computation $(q_i, s_i)_{i = 0}^{\infty}$ starting at $(q, s)$,
referred to as \df{the} computation in the text below.
This computation can be thought of as a run of $\A$ on an input tape with an
infinite sequence $a^\omega$.
\df{The} computation of $\A$ is,
naturally, the computation starting from $(q_0, \bot)$.
Note that it is due to the fact that $\A$
is unary that we are able to feed it a single infinite word instead of
countably many finite words.

In the text below we shall use the following notation and conventions.
To refer to an SLP $(S, \Terminals, \Nonterminals, \Productions)$, we sometimes
just use its axiom, $S$. The generated word, $w$, is denoted by $\str S$ as usual.
Note that the set of terminals is often understood from the context
and the set of nonterminals is always the set of left-hand sides of productions.
This enables us to use the notation $\str S$ to refer to the word generated
by the implicitly defined SLP, whenever the set of productions is clear from
the context.

\vspace{-3ex}
\subsubsection*{Transcripts of computations and overview of the algorithm.}
Recall that our goal is to describe an algorithm that, given a udpda $\A$,
produces an indicator pair for $L(\A)$.
We first assemble some tools that will allow us to handle
the computation of~$\A$ per~se. To this end, we introduce transcripts of
computations, which record ``events'' that determine whether certain input
words are accepted or rejected.

Consider a (finite or infinite) computation that consists of moves
$(q_i, s_i) \confto_{\sigma_i} \! (q_{i + 1}, s_{i + 1})$,
for $1 \le i \le k$ or for $i \ge 1$, respectively.
We define the \df{transcript} of such a computation as a (finite or infinite) sequence
\begin{equation*}
\mu(q_1) \, \sigma_1 \, \mu(q_2) \, \sigma_2 \, \ldots \, \mu(q_k) \, \sigma_k
\quad\text{or}\quad
\mu(q_1) \, \sigma_1 \, \mu(q_2) \, \sigma_2 \, \ldots,
\quad\text{respectively,}
\end{equation*}
where, for any $q_i$, $\mu(q_i) = \fsym$ if $q_i \in F$
and $\mu(q_i) = \eps$ if $q_i \in Q \setminus F$. Note that in the finite case
the transcript does \emph{not} include $\mu(q_{k + 1})$ where $q_{k + 1}$ is
the control state in the last configuration.
In particular, if a computation consists of a single configuration, then
its transcript is $\eps$.
In general, transcripts are finite words and infinite sequences over the auxiliary alphabet $\HistAlph$.

The reader may notice that our definition for the finite case
basically treats finite computations as left-closed, right-open intervals
and lets us perform their concatenation in a natural way. We note, however,
that from a technical point of view, a definition treating them
as closed intervals would actually do just as well.

Note that any sequence $s \in \HistAlph^\omega$ containing infinitely
many occurrences of $a$ naturally \df{defines} a unique characteristic sequence $c \in \Bin^\omega$
such that if $s$ is the transcript of a udpda computation,
then $c$ is the characteristic sequence of this udpda's language.
The following lemma shows that this correspondence is efficient
if the sequences are represented by pairs of SLPs.

\begin{lemma}
\label{l:transtochar}
There exists a polynomial-time algorithm that, given
a pair of straight-line programs $\TransPair$ of size $m$
that generates a sequence $s \in \HistAlph^\omega$ and
such that the symbol $a$ occurs in $\str \TransLoop$,
produces a pair of straight-line programs $\Pair$ of size $O(m)$
that generates the characteristic sequence defined by $s$.
\end{lemma}

\begin{proof}
Observe that it suffices to apply to the sequence generated by $\TransPair$
the composition of the following substitutions:
$h_1 \colon a \fsym \mapsto 1$,
$h_2 \colon a \mapsto 0$, and
$h_3 \colon \fsym \mapsto \eps$.
One can easily see that applying $h_2$ and $h_3$ reduces to applying
them to terminal symbols in SLPs, so it suffices to show that the application
of $h_1$ can also be done in polynomial time and increases the number of
productions in Chomsky normal form by at most a constant factor.

We first show how to apply $h_1$ to a single SLP.
Assume Chomsky normal form and process the
productions of the SLP inductively in the bottom-up direction.
Productions with terminal symbols remain unchanged, and
productions of the form $N \prodto A B$ are handled as follows:
if $\str A$ ends with an $a$ and $\str B$ begins with an $\fsym$, then
replace the production with $N \prodto (A a^{-1}) \cdot 1 \cdot (\fsym^{-1} B)$,
otherwise leave it unchanged as well. Here we use auxiliary nonterminals
of the form $N a^{-1}$ and $\fsym^{-1} N$ with the property that
$\str{N a^{-1}} \cdot a = \str N$ and $\fsym \cdot \str{\fsym^{-1} N} = \str N$.
These nonterminals are easily defined inductively in a straightforward manner,
just after processing $N$. At the end of this process one obtains an SLP
that generates the result of applying $h_1$ to the word generated by the original SLP.

We now show how to handle the fact that we need to apply $h_1$ to the entire
sequence $\str \TransPref \cdot (\str \TransLoop)^\omega$.
Process the SLPs $\TransPref$ and $\TransLoop$ as described above;
for convenience, we shall use the same two names for the obtained programs.
Then deal with the junction points in the
sequence $\str\TransPref \cdot (\str\TransLoop)^\omega$ as follows.
If $\str\TransLoop$ does not start with an $\fsym$, then there is nothing to
do. Now suppose it does; then there are two options. The first option
is that $\str\TransLoop$ ends with an $a$. In this case replace $\TransLoop$
with $(\fsym^{-1} \TransLoop) \cdot 1$ and $\TransPref$ with $(\TransPref a^{-1}) \cdot 1$
or with $(\TransPref \fsym)$ according to whether it ends with an $a$ or not.
The second option is that $\TransLoop$ does not end with an $a$. In this
case, if $\TransPref$ ends with an $a$, replace it with
$(\TransPref a^{-1}) \cdot 1 \cdot (\fsym^{-1} \TransLoop)$, otherwise
do nothing. One can easily see that the pair of SLPs obtained on this step will
generate the image of the original sequence $\str \TransPref \cdot (\str \TransLoop)^\omega$
under $h_1$. This completes the proof.
\qed
\end{proof}

\noindent
Note that
we could use a result by Bertoni, Choffrut, and Radicioni~\cite{bcr}
and apply a four-state transducer (however, the underlying automaton needs to
be $\eps$-free, which would make us figure out the last position ``manually'').

Now it remains to show how to efficiently produce, given a udpda $\A$,
a pair of SLPs $\TransPair$ generating the transcript of the computation
of $\A$. This is the key part of the entire algorithm, captured
by the following lemma.

\begin{lemma}
\label{l:trans}
There exists a polynomial-time algorithm that, given a udpda $\A$ of size $m$, produces a pair
of straight-line programs $\TransPair$ of size $O(m)$ that generates the transcript
of the computation of $\A$.
\end{lemma}

\noindent
The proof of Lemma~\ref{l:trans} is given in the next subsection.
Put together, Lemmas~\ref{l:transtochar} and~\ref{l:trans} prove
the harder direction (that is, part~\ref{th:ip:pdatoip}) of Theorem~\ref{th:ip}.
The only caveat is that if $\str\TransLoop \in \{ \fsym \}^*$, then one
needs to replace $\TransLoop$ with a simple SLP that generates $a$
and possibly adjust $\TransPref$ so that $\fsym$ be appended to the generated
word. This corresponds to the case where $\A$ does not read the entire input
and enters an infinite loop of $\eps$-moves (that is, moves that do not consume $a$
from the input).

\subsection{Details: proof of Lemma~\ref{l:trans}}
\label{ss:proof:translemma}

\subsubsection*{Returning and non-returning states.}
The main difficulty in proving Lemma~\ref{l:trans}
lies in capturing the structure of a unary deterministic computation.
To reason about such computations in a convenient manner,
we introduce the following definitions.

We say that a state $q$ is \df{returning} if it holds that
$(q, \bot) \confto^* (q', \bot)$ for some state $q' \in \Qpop$
(recall that states from $\Qpop$ are pop states).
In such a case the control state $q'$ of the first configuration
of the form $(q', \bot)$, $q' \in \Qpop$, occurring in the infinite
computation starting from $(q, \bot)$ is called the \df{exit point} of $q$,
and the computation between $(q, \bot)$ and this $(q', \bot)$
the \df{return segment} from $q$.
For example, if $q \in \Qpop$, then $q$ is its own exit point,
and the return segment from $q$ contains no moves.

Intuitively, the exit point is the first control state in the computation
where the bottom-of-the-stack symbol in the configuration $(q, \bot)$
may matter. One can formally show that
if $q'$ is the exit point of $q$, then for any configuration $(q, s)$
it holds that $(q, s) \confto^* (q', s)$ and, moreover, the transcript
of the return segment from $q$ is equal, for any $s$, to the transcript of the shortest
computation from $(q, s)$ to $(q', s)$.

If a control state is not returning, it is called \df{non-returning}.
For such a state $q$, it holds that for every configuration $(q', s')$
reachable from $(q, \bot)$ either $s' \ne \bot$ or $q' \not\in \Qpop$.
One can show formally that infinite computations starting from configurations
$(q, s)$ with a fixed non-returning state $q$ and arbitrary $s$
have identical transcripts and, therefore, identical characteristic sequences
associated with them. As a result, we can talk about infinite computations
starting at a non-returning control state $q$, rather than in a specific configuration $(q, s)$.

Now consider a state $q \not\in \Qpop$, an arbitrary configuration $(q, s)$
and the infinite computation starting from $(q, s)$.
Suppose that this computation enters a configuration of the form $(\bar q, s)$
after at least one move.
Then the \df{horizontal successor} of $q$ is defined as the control
state $\bar q$ of the first such configuration, and the computation
between these configurations is called the \df{horizontal segment} from $q$.
In other cases, horizontal successor and horizontal segment are undefined.
It is easily seen that the horizontal successor, whenever it exists, is well-defined
in the sense that it does not depend upon the choice
of $s \in (\StackAlph \setminus \{\bot\})^* \bot$. Similarly, the choice of $s$
only determines the ``lower'' part of the stack in the configurations of the
horizontal segment; since we shall only be interested in the transcripts, this
abuse of terminology is harmless.

Equivalently, suppose that $q \not\in \Qpop$ and $(q, s) \confto (q', s')$.
If $s' = s$ then the horizontal successor of $q$ is $q'$. Otherwise
it holds that $\pushes{q}{q'}{\gamma}$ for some $\gamma \in \StackAlph$, so that
$s' = \gamma s$. Now if $q'$ is returning, $q''$ is the exit point of $q'$,
and $\pops{q''}{\bar q}{\gamma}$ for the same $\gamma$, then $\bar q$ is the
horizontal successor of $q$. The horizontal segment is in both cases
defined as the shortest non-empty computation
of the form $(q, s) \confto^* (\bar q, s)$.%
\mynote{Diagrams in a later version?}

\vspace{-3ex}
\subsubsection*{General approach and data structures.}
Recall that our goal in this subsection is to define an algorithm
that constructs a pair of straight-line programs $\TransPair$
generating the transcript of the infinite computation of $\A$.
The approach that we take is dynamic programming.
We separate out intermediate goals of several kinds
and construct, for an arbitrary control state $q \in Q$,
SLPs and pairs of SLPs that generate transcripts
of the infinite computation starting at $q$ (if $q$ is non-returning),
of the return segment from $q$ (if $q$ is returning), and
of the horizontal segment from $q$ (whenever it is defined).

Our algorithm will \df{write} productions as it runs, always using,
on their right-hand side, only terminal symbols from $\HistAlph$ and
nonterminals defined by productions written earlier.
This enables us to use the notation $\str A$ for nonterminals $A$
without referring to a specific SLP.
Once written, a production is never modified.

The main data structures of the algorithm, apart from the productions
it writes, are as follows: three partial functions $\fE, \fH, \fW \colon Q \to Q$ and
a subset $\NonRet \sset Q$. Associated with $\fE$ and $\fH$ are nonterminals
$E_q$ and $H_q$, and with $\NonRet$ nonterminals $N'_q$ and $N''_q$.

Note that the partial functions from $Q$ to $Q$ can be thought of as digraphs on the set of
vertices $Q$. In such digraphs the outdegree of every vertex is at
most~$1$. The algorithm will subsequently modify these partial functions, that is,
add new edges and/or remove existing ones (however, the outdegree of no vertex
will ever be increased to above $1$).
We can also promise that $\fE$ will only increase, i.\,e., its graph will only
get new edges, $\fW$ will only decrease, and $\fH$ can go both ways.

During its run the algorithm will maintain the following invariants:

\begin{enumerate}

\renewcommand{\theenumi}{\textup{(}I\arabic{enumi}\textup{)}}
\renewcommand{\labelenumi}{\theenumi}

\item
\label{i:part}
{}$Q = \dom \fE \sqcup \dom \fH \sqcup \dom \fW \sqcup \NonRet$,
where $\sqcup$ denotes union of disjoint sets.

\item
\label{i:exit}
Whenever $\fE(q) = q'$, it holds that $q$ is returning, $q'$ is the exit point of $q$,
and $\str{E_q}$ is the transcript of the return segment from $q$.

\item
\label{i:horiz}
Whenever $\fH(q) = q'$, it holds that $q'$ is the horizontal successor of $q$ and
$\str{H_q}$ is the transcript of the horizontal segment from $q$.

\item
\label{i:push}
Whenever $\fW(q) = q'$, it holds that $\pushes{q}{q'}{\gamma}$
for some $\gamma \in \StackAlph$.

\item
\label{i:nonret}
Whenever $q \in \NonRet$, it holds that $q$ is non-returning
and the sequence $\str{N'_q} \cdot (\str{N''_q})^\omega$ is
the transcript of the infinite computation starting at $q$.

\end{enumerate}
\newcommand{\allinvs}{\ref{i:part}--\ref{i:nonret}}

\vspace{-3ex}
\subsubsection*{Description of the algorithm: computing transcripts.}
Our algorithm has three stages: the initialization stage, the main stage,
and the $\bot$-handling stage.
The initialization stage of the algorithm works as follows:

\begin{itemize}

\renewcommand{\labelitemi}{---}

\item for each $q \in Q$, write $V_q \prodto \mu(q) \sigma(q)$,
where $\mu(q)$ is $\fsym$ if $q \in F$ and $\eps$ otherwise,
and $\sigma(q)$ is $a$ if $q \in R$ (that is, if transitions
departing from~$q$ read a symbol from the input) and $\eps$ otherwise;

\item for all $q \in \Qpop$, set $\fE(q) = q$ and write $E_q \prodto \eps$;

\item for all $q \in \Qint$, set $\fH(q) = q'$ where $\ints{q}{q'}$ and write $H_q \prodto V_q$;

\item for all $q \in \Qpush$, set $\fW(q) = q'$ where $\pushes{q}{q'}{\gamma}$ for some $\gamma \in \StackAlph$;

\item set $\NonRet = \emptyset$.

\end{itemize}

\noindent
It is easy to see that in this way all invariants~\allinvs{} are initially satisfied
(recall that the transcript of an empty computation is $\eps$).


For convenience, we also introduce two auxiliary objects:
a partial function $\fG \colon Q \to Q$ and nonterminals $G_q$,
defined as follows.
The domain of $\fG$ is $\dom \fG = \dom \fH \sqcup \dom \fW$;
note that, according to invariant~\ref{i:part}, this union is disjoint.
We assign $\fG(q) = q'$ iff $\fH(q) = q'$ or $\fW(q) = q'$.
We shall assume that $\fG$ is recomputed as $\fH$ and $\fW$ change.
Now for every $q \in \dom \fG$, we let $G_q$ stand for $H_q$ if $q \in \dom \fH$
and for $V_q$ if $q \in \dom \fW$.

At this point we are ready to describe the main stage of the algorithm.
During this stage, the algorithm applies the following rules
until none of them is applicable (if at some point several rules
can be applied, the choice is made arbitrarily; the rules are
well-defined whenever invariants~\allinvs{} hold):

\begin{enumerate}

\renewcommand{\theenumi}{\textup{(}R\arabic{enumi}\textup{)}}
\renewcommand{\labelenumi}{\theenumi}

\item
\label{r:nonret}
If $\fG(q) = q'$ where $q' \in \NonRet$ and $q \in Q$:
remove $q$ from either $\dom \fH$ or $\dom \fW$,
add $q$ to $\NonRet$,
write $N'_q \prodto G_q N'_{q'}$ and $N''_q \prodto N''_{q'}$.

\item
\label{r:add}
If $\fH(q) = q'$ where $q' \in \dom \fE$ and $q \in Q$:
remove $q$ from $\dom \fH$,
define $\fE(q) = \fE(q')$,
write $E_q \prodto H_q E_{q'}$.

\item
\label{r:compose}
If $\fW(q) = q'$ where $q' \in \dom \fE$ and $q \in Q$:
remove $q$ from $\dom \fW$,
define $\fH(q) = \bar q$ where
$\fE(q') = q''$, $\pops{q''}{\bar q}{\gamma}$,
and $\pushes{q}{q'}{\gamma}$ (that is, $\gamma$ is the symbol
pushed by the transition leaving $q$, and $\bar q$ is
the state reached by popping $\gamma$ at $q''$, the exit point of $q'$).
Finally, write $H_q \prodto V_q E_{q'} V_{q''}$.

\item
\label{r:cycle}
If $\fG$ contains a simple cycle, that is,
if $\fG(q_i) = q_{i + 1}$ for $i = 1, \ldots, k - 1$
and $\fG(q_k) = q_1$, where $q_i \ne q_j$ for $i \ne j$,
then for each vertex $q_i$ of the cycle
remove it from either $\dom \fH$ or $\dom \fW$ and
add it to $\NonRet$;
in addition, write $N'_{q_k} \prodto G_{q_k}$,
$N''_{q_k} \prodto G_{q_1} \ldots G_{q_k}$,
and, for each $i = 1, \ldots, k - 1$,
$N'_{q_i} \prodto G_{q_i} N'_{q_{i + 1}}$
and $N''_{q_i} \prodto N''_{q_{i + 1}}$.

\end{enumerate}
\newcommand{\allrules}{\ref{r:nonret}--\ref{r:cycle}}

\noindent
We shall need two basic facts about this stage of the algorithm.

\begin{myclaim}
\label{c:invhold}
Application of rules~\allrules{} does not violate invariants~\allinvs.
\end{myclaim}

\noindent
The proof of Claim~\ref{c:invhold} is easy and left to the reader.

\begin{myclaim}
\label{c:term}
If no rule is applicable, then $\dom \fG = \emptyset$.
\end{myclaim}

\begin{proof}
Suppose $\dom \fG \ne \emptyset$. Consider the graph associated with $\fG$
and observe that all vertices in $\dom \fG$ have outdegree~$1$.
This implies that $\fG$ has either a cycle within $\dom \fG$ or an edge
from $\dom \fG$ to $Q \setminus \dom \fG$. In the first case,
rule~\ref{r:cycle} is applicable. In the second case,
we conclude with the help of the invariant~\ref{i:part} that
the edge leads from a vertex in $\dom \fH \sqcup \dom \fW$
to a vertex in $\NonRet \sqcup \dom \fE$. If the destination is
in $\NonRet$, then rule~\ref{r:nonret} is applicable;
otherwise the destination is in $\dom \fE$ and one can apply
rule~\ref{r:add} or rule~\ref{r:compose} according to
whether the source is in $\dom \fH$ or in $\dom \fW$.
\qed
\end{proof}

\noindent
Now we are ready to describe the $\bot$-handling stage of the algorithm.
By the beginning of this stage, the structure of deterministic computation
of $\A$ has already been almost completely captured by the productions written
earlier, and it only remains to account for moves involving $\bot$.
So this last stage of the algorithm takes the initial state $q_0$ of $\A$ and proceeds as follows.

If $q_0 \in \NonRet$, then take $N'_{q_0}$ as the axiom of $\TransPref$
and $N''_{q_0}$ as the axiom of $\TransLoop$. By invariant~\ref{i:nonret},
these nonterminals are defined and generate appropriate words,
so the pair $\TransPair$ indeed generates the transcript of the computation
of $\A$.

Since at the beginning of the $\bot$-handling stage
$\dom \fG = \emptyset$, it remains to consider the case $q_0 \in \dom \fE$.
Define a partial function $\fEbot \colon Q \to Q$ by setting,
for each $q \in \dom \fE$, its value according to $\fEbot(q) = \bar q$
if $\fE(q) = q'$ and $\pops{q'}{\bar q}{\bot}$. Write productions
$\Ebot_q \prodto E_q V_{q'}$ accordingly. Now associate $\fEbot$
with a graph, as earlier, and consider the longest simple path within $\dom \fEbot$ starting at
$q_0$. Suppose it ends at a vertex $q_k$, where $\fEbot(q_i) = q_{i + 1}$
for $i = 0, \ldots, k$. There are two subcases here according to why
the path cannot go any further.

The first possible reason is that it reaches $Q \setminus \dom \fEbot = \NonRet$,
that is, that $q_{k + 1}$ belongs to $\NonRet$. In this subcase write
$N'_{q_0} \prodto \Ebot_{q_0} \ldots \Ebot_{q_{k}} N'_{q_{k + 1}}$
and $N''_{q_0} \prodto N''_{q_{k + 1}}$.
The second possible reason is that $q_{k + 1} = q_i$ where $0 \le i \le k$.
In this subcase write $N'_{q_0} \prodto \Ebot_{q_0} \ldots \Ebot_{q_{i - 1}}$
and $N''_{q_0} \prodto \Ebot_{q_i} \ldots \Ebot_{q_k}$.

In any of the two subcases above, take $N'_{q_0}$ and $N''_{q_0}$ as axioms
of $\TransPref$ and $\TransLoop$, respectively. The correctness of this step follows
easily from the invariants~\ref{i:exit} and~\ref{i:nonret}.
This gives a polynomial algorithm that converts a udpda $\A$
into a pair of SLPs $\TransPair$ that generates
the transcript of the infinite computation of $\A$,
and the only remaining bit is bounding the size of $\TransPair$.

\begin{myclaim}
The size of $\TransPair$
is $O(|Q|)$.
\end{myclaim}

\begin{proof}
There are three types of nonterminals whose productions may have growing size:
$N''_{q_k}$ in rule~\ref{r:cycle}, and $N'_{q_0}$ and $N''_{q_0}$ in the $\bot$-handling
stage. For all three types, the size is bounded by the cardinality of the set of states involved in a cycle
or a path. Since such sets never intersect, all such nonterminals together
contribute at most $|Q|$ productions to the Chomsky normal form. The contribution
of other nonterminals is also $O(|Q|)$, because they all have fixed size and
each state $q$ is associated with a bounded number of them.
\qed
\end{proof}

\noindent
Combined with Claim~\ref{c:udpdanf} in Subsection~\ref{ss:proof:pre},
this completes the proof of Lemma~\ref{l:trans} and Theorem~\ref{th:ip}.

\section{Universality of unpda}
\label{s:nondet}

In this section we settle the complexity status of the universality problem
for unary, possibly nondeterministic pushdown automata.
While $\PiTwoP$-completeness of equivalence and inclusion is
shown by Huynh~\cite{huynh}, it has been unknown whether the universality
problem is also $\PiTwoP$-hard.

For convenience of notation, we use an auxiliary descriptional system.
Define \df{integer expressions} over the set of operations $\intexpop$
inductively: the base case is a non-negative integer $n$, written
in binary, and the inductive step is associated with binary operations
$+$, $\intor$, and unary operations $\intdoubled$,~$\intiter{}$.
To each expression $E$ we associate
a set of non-negative integers $\intsem E$: $\intsem{n} = \{ n \}$,
$\intsem{E_1 + E_2} = \{ s_1 + s_2 \colon s_1 \in \intsem{E_1}, s_2 \in \intsem{E_2} \}$,
$\intsem{E_1 \intor E_2} = \intsem{E_1} \cup \intsem{E_2}$,
$\intsem{E \intdoubled} = \intsem{E + E}$,
$\intsem{\intiter E} = \{ s k \colon s \in \intsem E, k = 0, 1, 2, \ldots \,\}$.

Expressions $E_1$ and $E_2$ are called \df{equivalent} iff $\intsem{E_1} = \intsem{E_2}$;
an expression $E$ is \df{universal} iff it is equivalent to $\intiter 1$.
The problem of deciding universality is denoted by \IntExpUnivers.

Decision problems for integer expressions have been studied for more than 40~years:
Stockmeyer and Meyer~\cite{sm73} showed that for expressions over $\{+,\intor\}$
compressed membership is $\NP$-complete and equivalence is $\PiTwoP$-complete
(universality is, of course, trivial). For recent results on such problems
with operations from $\{+, \cup, \cap, \times, \overline{\phantom{x}}\}$,
see McKenzie and Wagner~\cite{mw07} and Gla\ss{}er et~al.~\cite{ghrtw}.

\begin{lemma}
\label{l:intexp}
\IntExpUnivers\ is $\PiTwoP$-hard.
\end{lemma}

\begin{proof}
The reduction is from the \GenSubsetSum\ problem,
which is defined as follows. The input consists of two vectors of naturals,
$u = (u_1, \ldots, u_n)$ and $v = (v_1, \ldots, v_m)$, and a natural $t$,
and the problem is to decide whether for all $y \in \Bin^m$ there exists
an $x \in \Bin^n$ such that $x \conv u + y \conv v = t$, where the middle dot $\conv$ once
again denotes the inner product. This problem was shown to be hard by Berman
et~al.~\cite[Lemma~6.2]{bklpr}.\mynote{Reference to the compendium: in the intro or here.}
\par
Start with an instance of \GenSubsetSum\ and
let $M$ be a big number, $M > \sum_{i = 1}^{n} u_i + \sum_{j = 1}^{m} v_j$.
Assume without loss of generality that $M > t$.
Consider the integer expression $E$ defined by the following equations:
\begin{align*}
E &= E' \intor E'', \\
E' &= (2^m M + \intiter 1) \intor (\intiter M + ([0, t - 1] \intor [t + 1, M - 1])),\\
E'' &= \sum_{j = 1}^{m} (0 \intor (2^{j - 1} M + v_j)) +
       \sum_{i = 1}^{n} (0 \intor u_i), \\
[a, b] &= a + [0, b - a], \\
[0, t] &= [0, \lfloor t/2 \rfloor] \intdoubled + (0 \intor (t \bmod 2)), \\
[0, 1] &= 0 \intor 1, \\
[0, 0] &= 0.
\end{align*}
Note that the size of $E$ is polynomial in the size of the input,
and $E$ can be constructed in logarithmic space.
We show that $E$ is universal iff the input is a yes-instance
of \GenSubsetSum.
\par
It is immediate that $E$ is universal if and only if $\intsem E$ contains $2^m$
numbers of the form $k M + t$, $0 \le k < 2^m$.
We show that every such number is in $\intsem E$ if and only if for
the binary vector $y = (y_1, \ldots, y_m) \in \Bin^m$,
defined by $k = \sum_{j = 1}^{m} y_j \, 2^{j - 1}$,
there exists a vector $x \in \Bin^n$ such that $x \conv u + y \conv v = t$.
\par
First consider an arbitrary $y \in \Bin^m$ and
choose $k$ as above.
Suppose that for this $y$ there exists an $x \in \Bin^n$ such that
$x \conv u + y \conv v = t$. One can easily see that appropriate choices
in $E''$ give the number $k M + y \conv v + x \conv u = k M + t$.
Conversely, suppose that $k M + t \in \intsem E$ for some $k$, $0 \le k < 2^m$;
then $k M + t \in \intsem{E''}$. Since $(1, \ldots, 1) \conv u + (1, \ldots, 1) \conv v < M$,
it holds that $t = y \conv v + x \conv u$ for binary vectors $y \in \Bin^m$
and $x \in \Bin^n$ that correspond to the choices in the addends.
Moreover, the same inequality also shows that $k M$ is equal to
the sum of some powers of two in the first sum in $E''$, and so
$k = \sum_{j = 1}^{m} y_j \, 2^{j - 1}$.
This completes the proof.
\qed
\end{proof}

\begin{remark*}
With \emph{circuits} instead of formulae (see also~\cite{mw07} and~\cite{ghrtw}) we would not need doubling.
Furthermore, we only use $\intiter{}$ on fixed numbers, so instead we could use
any feature for expressing an arithmetic progression with fixed
common difference.
\end{remark*}

\begin{theorem}
\label{th:nondet}
\NondetUnivers\ is $\PiTwoP$-complete.
\end{theorem}

\begin{proof}
A reduction from \IntExpUnivers, which is $\PiTwoP$-hard
by Lemma~\ref{l:intexp}, shows hardness.
Indeed, an integer expression over $\intexpop$ can be transformed
into a unary CFG in a straightforward way. Binary numbers are encoded
by poly-size SLPs, summation is modeled by concatenation, and union by alternatives.
Doubling is a special case of summation, and $\intiter{}$ gives rise to
productions of the form $N' \prodto \eps$ and $N' \prodto N N'$. The obtained
CFG is then transformed into a unary PDA $\A$ by a standard algorithm
(see, e.\,g.,
Savage~\cite[Theorem~4.12.1]{savage}).
The result is that $L(\A) = \{ 1^s \colon s \in \intsem E \}$, and $\A$ is computed from $E$ in
logarithmic space.
This concludes the proof.
\qed
\end{proof}


\begin{remark*}
We give a simple proof of the $\PiTwoP$ upper bound.
Let $\varphi_{\A}(x)$ be an existential Presburger formula of size
polynomial in the size of $\A$ that characterizes the Parikh image of $L(\A)$
(see Verma, Seidl, and Schwentick~\cite[Theorem~4]{vss}).
To show that an udpda $\A$ is non-universal, we find an $n\geq 0$
such that $\lnot \varphi_{\A}(n)$ holds.
Now we note that for any udpda $\A$ of size $m$, there is a deterministic finite
automaton
of size $2^{O(m)}$ accepting $L(\A)$ (see discussion in Section~\ref{s:disc} and Pighizzini~\cite{pudpda}).
Thus, $n$ is bounded by $2^{O(m)}$.
Therefore, checking non-universality can be expressed as a
predicate:
$\exists n \leq 2^{O(m)}.\lnot \varphi_{\A}(n)$.
This is a $\SigmaTwoP$-predicate, because the $\exists^*$-fragment
of Presburger arithmetic is $\NP$-complete~\cite{vzgs78}.
\end{remark*}

\begin{corollary}
\label{cor:nondet}
Universality, equivalence, and inclusion are $\PiTwoP$-complete
for (possibly nondeterministic) unary pushdown automata, unary context-free grammars,
and integer expressions over $\intexpop$.
\end{corollary}

\noindent
Another consequence of Theorem~\ref{th:nondet} is that
deciding equality of a (not necessarily unary) context-free language,
given as a context-free grammar, to any fixed context-free language $L_0$
that contains an infinite regular subset, is $\PiTwoP$-hard and,
if $L_0 \sset \aalph^*$, $\PiTwoP$-complete.
The lower bound is by reduction
due to Hunt~III, Rosenkrantz, and Szymanski~\cite[Theorem~3.8]{hrs},
who show that deciding equivalence to $\aalph^*$ reduces to deciding
equivalence to any such $L_0$. The reduction is shown to be polynomial-time,
but is easily seen to be logarithmic-space as well.
The upper bound for the unary case is by Huynh~\cite{huynh};
in the general case, the problem can be undecidable.

\section{Corollaries and discussion}
\label{s:disc}

\subsubsection*{Descriptional complexity aspects of udpda.}

Theorem~\ref{th:ip} can be used to obtain
several results on descriptional complexity aspects of udpda
proved earlier by Pighizzini~\cite{pudpda}.
He shows how to transform a udpda of size $m$
into an equivalent deterministic finite automaton (DFA) with at most $2^m$ states~\cite[Theorem~8]{pudpda}
and into an equivalent context-free grammar in Chomsky normal form (CNF) with
at most $2 m + 1$ nonterminals~\cite[Theorem~12]{pudpda}.
In our construction $m$ gets multiplied by a small constant, but the advantage is
that we now see (the slightly weaker variants of) these results as easy
corollaries of a single underlying theorem.
Indeed, using
an indicator pair $\Pair$ for $L$, it is straightforward to
construct a DFA of size $|\str\Pref| + |\str\Loop|$ accepting $L$,
as well as to transform the pair into a CFG in CNF that generates $L$ and
has at most thrice the size of $\Pair$.

Another result which follows, even more directly, from ours
is a lower bound on the size of udpda accepting a specific
language~$L_1$~\cite[Theorem~15]{pudpda}.
%
%
To obtain this lower bound,
Pighizzini employs a known lower bound on the SLP-size of the word
$W = W\![0] \ldots W\![K - 1] \in \Bin^K$ such that
$a^n \in L_1$ iff $W\![n \bmod K] = 1$. To this end, a udpda $\A$ accepting $L_1$
is intersected (we are glossing over some technicalities here) with a small
deterministic finite automaton that ``captures'' the end of the word $W$.
The obtained udpda, which only accepts $a^K$, is transformed into an equivalent
context-free grammar. It is then possible to use the structure of the grammar
to transform it into an SLP that produces $W$ (note that such a transformation
in general is $\NP$-hard).
While the proof produces from a udpda for $L_1$ a related SLP with a polynomial blowup,
this construction depends crucially on the structure of the language $L_1$,
so it is difficult to generalize the
argument to \emph{all} udpda and thus obtain Theorem~\ref{th:ip}.
Our proof of Theorem~\ref{th:ip} therefore follows a very different path.


\vspace{-3ex}
\subsubsection*{Relationship to Presburger arithmetic.}
An alternative way to prove the upper bound in Theorem~\ref{th:incl} 
is via Presburger arithmetic, using the observation that there is a poly-time computable
existential Presburger formula
that expresses the membership of a word $a^n$ in $L(\lnot \A_1)$ and $L(\A_2)$.
This technique distills the arguments
used by Huynh~\cite{huynhcomm,huynh} to show that the compressed membership
problem for unary pushdown automata is in $\NP$.
It is used in a purified form by Plandowski and Rytter~\cite[Theorems~4 and~8]{prjewels},
who developed a much shorter proof of the same fact (apparently unaware of the previous proof).
The same idea was later rediscovered and used in a combination with Presburger arithmetic
by Verma, Seidl, and Schwentick~\cite[Theorem~4]{vss}.

Another application of this technique provides
an alternative proof of the $\PiTwoP$ upper bound
for unpda inclusion (Theorem~\ref{th:nondet}):
to show that $L(\A)$ is universal, we check that $L(\A)$ accepts all words
up to length $2^{O(m)}$ (this bound is sufficient because there is a deterministic
finite automaton for the language with this size---see the discussion above).
The proof known to date, due to Huynh~\cite{huynh},
involves reproving Parikh's theorem and is more than 10 pages long.
Reduction to Presburger formulae produces a much simpler proof.

Also, our $\PiTwoP$-hardness result for unpda shows that 
the $\forall_{\mathrm{bounded}}\,\exists^*$-fragment of
Presburger arithmetic is $\PiTwoP$-complete, where the variable bound by the universal
quantifier is at most exponential in the size of the formula.
The upper bound holds because the $\exists^*$-fragment is
$\NP$-complete~\cite{vzgs78}.
In comparison,
the $\forall\, \exists^*$-fragment, without any restrictions
on the domain of the universally quantified variable, requires
co-nondeterministic $2^{n^{\MyOmega(1)}}$ time, see Gr\"adel~\cite{g89}.
Previously known fragments that are complete for the second level of the polynomial hierarchy
involve alternation depth~$3$ and a fixed number of quantifiers, as in
Gr\"adel~\cite{g88} and Sch\"oning~\cite{sch97}.
Also note that the $\forall^s \, \exists^t$-fragment
is $\coNP$-complete for all fixed $s \ge 1$ and $t \ge 2$, see Gr\"adel~\cite{g88}.

\vspace{-3ex}
\subsubsection*{Problems involving compressed words.}
Recall Theorem~\ref{th:comporder}: given two SLPs,
it is $\coNP$-complete to compare the generated words componentwise
with respect to any partial order different from equality.
As a corollary,
we get precise complexity bounds for
SLP equivalence in the presence of wildcards or, equivalently,
compressed matching in the model of partial words
(see, e.\,g., Fischer and Paterson~\cite{fpdontcares} and
Berstel and Boasson~\cite{bbpartial}).
Consider the problem \PartialMatch:
the input is a pair of SLPs $\Prog_1$, $\Prog_2$ over the alphabet $\{a, b, \qmark\}$,
generating words of equal length, and the output is ``yes'' iff
for every $i$, $0 \le i < |\Prog_1|$, either $\Prog_1[i] = \Prog_2[i]$
or at least one of $\Prog_1[i]$ and $\Prog_2[i]$ is $\qmark$ (a \df{hole},
or a single-character \df{wildcard}).


Schmidt-Schau\ss~\cite{ss12} defines a problem equivalent to \PartialMatch, along with
another related problem, where one needs to find occurrences of $\str{\Prog_1}$ in
$\str{\Prog_2}$ (as in pattern matching), $\Prog_2$ is known to contain no
holes, and two symbols match iff they are equal or at least one of
them is a hole. For this related problem, he develops
a poly\-nomial-time algorithm that finds (a representation of) all matching occurrences and operates
under the assumption that the number of holes in $\str{\Prog_1}$
is polynomial in the size of the input.
He also points out that no solution
for (the general case of) \PartialMatch\ is known---unless
a polynomial upper bound on the number of $\qmark$s in $\str{\Prog_1}$ and $\str{\Prog_2}$
is given.
Our next proposition shows that such a solution is not possible
unless $\P = \NP$. It is an easy consequence of Theorem~\ref{th:comporder}.

\begin{proposition}
\label{p:pm}
\PartialMatch\ is $\coNP$-complete.
\end{proposition}

\begin{proof}
Membership in $\coNP$ is obvious, and the
hardness is by a reduction from \CompSLPbin. Given a pair of SLPs $\Prog_1$, $\Prog_2$
over $\Bin$, substitute $\qmark$ for $0$ and $a$ for $1$ in $\Prog_1$,
and $b$ for $0$ and $\qmark$ for $1$ in $\Prog_2$. The resulting pair of SLPs
over $\{a, b, \qmark\}$ is a yes-instance of \PartialMatch\ iff the original
pair is a yes-instance of \CompSLPbin.
\qed
\end{proof}

\noindent
The wide class of \df{compressed membership} problems (deciding $\str\Prog \in L$)
is studied and discussed
in Je\.{z}~\cite{jezcomp} and Lohrey~\cite{l12}.
In the case of words over the unary alphabet, $w \in \aalph^*$, expressing $w$
with an SLP is poly-time equivalent to representing it with its length $|w|$
written in binary.
An easy corollary of Theorem~\ref{th:member} is that
deciding $w \in L(\A)$, where $\A$ is a (not necessarily unary) deterministic
pushdown automaton and $w = a^n$ with $n$ given in binary, is $\P$-complete.

Finally, we note that the precise complexity of SLP equivalence remains open~\cite{l12}.
We cannot immediately apply lower bounds for udpda equivalence, since we do not know if
the translation from udpda to indicator pairs in Theorem~\ref{th:ip} can be implemented
in logarithmic (or even polylogarithmic) space.

%

\end{document}